\documentclass[12pt, letterpaper]{article} 

\usepackage{amsfonts}
\usepackage{amssymb}
\usepackage{graphicx}
\usepackage{algorithm}
\usepackage{algpseudocode}
\usepackage[fleqn]{amsmath}
\usepackage{amsfonts}
\usepackage{amssymb}
\usepackage{appendix}
\usepackage{natbib}
\usepackage[colorlinks=true,linkcolor=blue,urlcolor=blue,citecolor=blue]{hyperref}

\usepackage{amsthm}
\usepackage{mathtools}

\newtheorem{definition}{Definition}
\newtheorem{remark}{Remark}
\newtheorem{lemma}{Lemma}
\newtheorem{theorem}{Theorem}
\newtheorem{assumption}{Assumption}
\newtheorem{corollary}{Corollary}

\newcommand\gss{\mathcal{G}}
\newcommand\ptt{p}
\newcommand\jcost{\mathbb{J}} 

\title{Model-free optimal controller for discrete-time Markovian jump linear systems: A Q-learning approach} 

\author{Ehsan Badfar, Babak Tavassoli\thanks{K.N. Toosi University of Technology\\ (ehsan.badfar@email.kntu.ac.ir, tavassoli@kntu.ac.ir).}}

\begin{document}
\maketitle

\begin{abstract}
This research paper introduces a model-free optimal controller for discrete-time Markovian jump linear systems (MJLSs), employing principles from the methodology of reinforcement learning (RL). While Q-learning methods have demonstrated efficacy in determining optimal controller gains for deterministic systems, their application to systems with Markovian switching remains unexplored. To address this research gap, we propose a Q-function involving the Markovian mode. Subsequently, a Q-learning algorithm is proposed to learn the unknown kernel matrix using raw input-state information from the system. Notably, the study proves the convergence of the proposed Q-learning optimal controller gains to the model-based optimal controller gains after proving the convergence of a value iteration algorithm as the first step. Addition of excitation noise to input which is required to ensure the leaning performance does not lead to any bias. Unlike the conventional optimal controller, the proposed method does not require any knowledge on system dynamics and eliminates the need for solving coupled algebraic Riccati equations arising in optimal control of MJLSs. Finally, the efficiency of the proposed method is demonstrated through a simulation study.\\[6pt]
\end{abstract}

\section{Introduction}
\begingroup
\setlength{\parskip}{.2 em}               
Markovian systems constitute a practically significant category within stochastic switching systems. They can be employed to model dynamics that switch abruptly across a finite set of linear or nonlinear models, guided by a finite Markov chain \citep{shi2015survey,zhao2019brief}. They are widely used in various domains, encompassing modeling, stability analysis, and control within real-world applications, such as networked control systems \citep{gupta2009networked,zhang2015survey}, power systems \citep{sanjari2016probabilistic,correa2021cooperative}, and cyber-attacks \citep{wu2016survey,cao2019finite}.

During the past decade, several efforts have been directed towards the advancement of control methodologies for MJLS. For instance, in \cite{tzortzis2020jump}, the authors designed a Linear Quadratic Regulator (LQR) for MJLS with uncertain transition probabilities through a min-max optimization approach. In \cite{han2021optimal}, the authors introduced a finite horizon optimal controller for MJLS with input delay with stability guarantee in mean-square sense. In \cite{tang2024improved}, the authors proposed a novel augmentation technique to develop criteria for the finite-time controllability of time-varying delayed Markovian jump control networks. Additionally, other approaches such as sliding mode control \citep{qi2019sliding}, model predictive control \citep{zhang2022model}, event-triggered control \citep{zhou2024event}, and fuzzy control \citep{tan2024robust} have been successfully applied to control MJLS.

All the previously mentioned research works rely on the model of system to achieve the controller gain for a Markovian jump linear system (MJLS). In other words, it is essential to have complete knowledge of system dynamics to derive the control input signal. Since the dynamics of the system might be unavailable in many practical applications, the design of model-free controllers can be regarded as an important challenge for the control engineering community \citep{annaswamy2023control,dorfler2023data}.

An innovative approach to solve this challenging problem lies in the adoption of machine learning techniques, particularly RL \citep{bertsekas2019reinforcement}. In RL, an intelligent agent exchanges data with an unknown environment. The intelligent agent mainly aims to learn the policies by observing the response from the unknown environment. Recently, RL has demonstrated outstanding results across diverse application domains, including robotics \citep{polydoros2017survey,singh2022reinforcement}, transportation \citep{cao2020using,wei2021recent}, and economics \citep{charpentier2021reinforcement,hambly2023recent}. From a control engineering perspective, RL has a high potential to bridge the gap between conventional optimal control and data-driven control \citep{kiumarsi2017optimal,bucsoniu2018reinforcement,annaswamy2023adaptive}. To mention some of the approaches to applying RL methods in various control problems, \cite{rizvi2018output} formulates a model-free $H_{\infty}$ controller for a class of linear discrete-time system based on the input-output information from the system. In \cite{fan2019model}, the authors develop a model-free tracking controller for discrete-time linear systems through an off-policy method. \cite{farjadnasab2022model} employs RL techniques to develop LQR for a decentralized system. \cite{jiang2023reinforcement} uses RL to design a cooperative $H_{\infty}$ output regulation for multi-agent systems. \cite{huang2024specified} utilizes RL methods to develop zero tracking error optimal controller for linear discrete-time systems with user-specified convergence rate.

In contrast to the recent advances of RL in developing model-free optimal controllers for deterministic systems, the problem of designing model-free optimal controllers for an MJLS with RL has not yet been investigated. The primary objective of this research is to employ RL methods, specifically Q-learning, to develop a model-free controller for discrete-time MJLS. This controller must converge to the conventional optimal controller for MJLS without the need for prior knowledge of system dynamics or the offline solution of the Algebraic Riccati equation backward in time. Accordingly, the contributions of this research paper are summarized as follows. We introduce a model-free optimal controller for a linear discrete-time system with Markovian transitions. For this purpose, a novel formulation of value function and Q-function are suggested. It is demonstrated that the excitation noise during the learning process does not introduce any bias in the estimation of the kernel matrix. Additionally, the algorithm proposed in this research does not require an initial stabilizing control policy. Finally, it is proven that the proposed model-free controller converges to the model-based optimal controller, which is guaranteed to be mean-square stable.

The rest of this research paper is organized as follows. Section~\ref{section:introduction} briefly presents the MJLS and the corresponding conventional optimal controller. In Section~\ref{section:ValueIterationAlgorithm}, the value iteration algorithm for MJLS is developed. Section~\ref{section:Q-Learning} formulates Q-learning method and algorithm for MJLS using the least squares (LS) technique to learn the unknown kernel matrix of the Q-function. Furthermore, the convergence of the proposed method to the conventional optimal controller is proved. Then, a simulation study is presented in Section~\ref{section:Simulation} to demonstrate the efficacy of the proposed controller. Finally, the conclusions are provided in Section~\ref{section:conclusion}.

\vspace{\baselineskip}
\noindent \textbf{Notation:} The following standard notations are used in this paper. The sets of non-negative integers and real numbers are denoted by \( \mathbb{Z}_0^+ \) and \(\mathbb{R}\), respectively. Also, \( \mathbb{R}^n \) and \( \mathbb{R}^{n \times m} \) denote the \( n \)-dimensional Euclidean space and the set of \( n \times m \) real matrices, respectively. \( I \) is an identity matrix with appropriate dimensions. A set of \( N \) matrices \( A_1, A_2, \ldots, A_N \in \mathbb{R}^{m\times n} \) is denoted as an \( N \)-matrix \( A = (A_1, A_2, \ldots, A_N) \in \mathbb{R}^{m\times n\times N}\). \( \lVert x \rVert \) is the 2-norm of vector \( x \). 
For events $A$ and $B$, the probability of \( A \) and the conditional probability of \( A \) given \( B \) are denoted by \(\mathbb{P}[A] \) and \( \mathbb{P}[A|B] \), respectively.
For random variables \( x \) and \( y \), \( \mathbb{E}[x] \) is the expected value of \( x \) and \( \mathbb{E}[x|y] \) is the conditional expectation of \( x \) given \( y \).  Matrices are assumed to be compatible for algebraic operations unless their dimensions are explicitly specified.
\endgroup
\section{Background and problem statement} \label{section:introduction}
\label{Section2}
In this section, we first introduce the Markovian system under study, together with the stability notation. Then, we briefly present the model-based optimal controller for MJLS.
\subsection{System description}
The discrete-time MJLS in the state-space form is considered as 
\begin{equation}
\gss: \begin{cases}
x_{k+1} = A_{\theta_k} x_k + B_{\theta_k} u_k \\
y_k = C_{\theta_k} x_k\label{eq:sys_dynamics}
\end{cases}
\end{equation}
where $x_k \in \mathbb{R}^n$, $y_k \in \mathbb{R}^q$, and $u_k \in \mathbb{R}^m$ are the state, output, and control input vectors, respectively. The system matrices are $N$-matrices given as $A = (A_1, \dots, A_N)$, $B = (B_1, \dots, B_N)$, and $C = (C_1, \dots, C_N)$ with compatible dimensions. 
The Markovian jump parameter $\theta_k$ takes values from the finite set $\Theta = \{1, 2, \dots, N\}$ with transition probabilities  
\begin{equation}
\label{eq:tpr}
\ptt_{ij} = \mathbb{P}[\theta_{k+1}=j \mid \theta_k=i],
\end{equation}
where $\ptt_{ij} \in [0,1]$ for all $i,j \in \Theta$ and $\sum_{j=1}^N \ptt_{ij}=1$ for all $i \in \Theta$.
The probability that the Markovian mode is $i \in \Theta$ at time step $k$ is denoted as 
\begin{equation}
\pi_{k,i} = \mathbb{P}[\theta_k=i].
\end{equation}
which satisfies $\sum_{i=1}^N \pi_{k,i} = 1$ for all $k\in\mathbb{Z}_0^+$.
The probability distribution of $\theta_k$ is given by the vector $\pi_k\in\mathbb{R}^N$ defined such that $[\pi_k]_i=\pi_{k,i}$. The transition probability matrix is obtained as $[\Phi]_{ij}=\ptt_{ij}$ for $i,j\in\Theta$. Then, one can use \eqref{eq:tpr} to write
\begin{equation}
\pi_{k+1}^T = \pi_k^T \Phi \quad \forall k\in\mathbb{Z}_0^+.
\end{equation} 

If the limit probability distribution $\pi_\infty = \lim_{k\to\infty} \pi_k$ exists, then the Markov chain $\{\theta_k\}_{k=0}^\infty$ is called ergodic \citep{costa2006discrete}.

The overall state of MJLS at time step \( k \) is composed of the pair of system state and Markovian mode which is denoted as 
\begin{equation}
I_k = ( x_k, \theta_k ).
\end{equation}

The stability of system $\gss$ is studied in the context of mean-square stability in \cite{costa2006discrete}, which can be represented as the following definition for the system $\gss$ in \eqref{eq:sys_dynamics}.

\begin{definition}
\label{defin2}
The system $\gss$ is mean-square stable if the following relation holds for all $I_0 \in \mathbb{R}^n \times \Theta$.
\begin{equation}
\sum_{k=0}^{\infty} \mathbb{E}\left[\left\| x_k \right\|^2 \mid I_0 \right] < \infty.
\end{equation}
\end{definition}

\begin{definition}[\citealt{zhao2008practical}] 
\label{defin3}
The pair $(A,B)$ is mean-square stabilizable if there exists a control law in the form of $u_k=-K_{\theta_k} x_k$ with $K=(K_1,\dots,K_N) \in \mathbb{R}^{m \times n}$, such that the closed-loop system is mean-square stable.
\end{definition}

\subsection{Model-based Optimal Controller}
In the conventional optimal control design for MJLS, it is desired to find a stabilizing control law that minimizes an infinite horizon quadratic cost function
\begin{equation}
\label{eq:cost}
\jcost\left(x_0, \theta_0\right)=\mathbb{E}\left[\sum_{k=0}^{\infty} x_k^T Q_{\theta_k} x_k+u_k^T R_{\theta_k} u_k \mid I_0\right] 
\end{equation}
for the initial state $I_0=(x_0,\theta_0)$. The weighting matrices \( Q=(Q_1,\ldots,Q_N) \) and \( R=(R_1,\ldots,R_N) \) are \( N \)-matrices with appropriate dimensions such that $R_i>0$ and $Q_i\ge 0$ for $i\in\Theta$.

To facilitate the control design it is common to make the following reasonable assumptions \citep{costa2006discrete}.

\begin{assumption}
\label{assum2}
It is assumed that the state variables and Markovian modes of the system $\gss$ in \eqref{eq:sys_dynamics} are conditionally independent random variables such that 
\begin{equation}
\begin{split}
&\mathbb{E}[f(x_{k+1}) g(\theta_{k+1}) \mid I_k] \\
&= \mathbb{E}[f(x_{k+1}) \mid I_k] \mathbb{E}[g(\theta_{k+1}) \mid I_k].
\end{split}
\end{equation}
\end{assumption}

\begin{assumption}
\label{assum3}
It is assumed that the Markovian mode of the system $\gss$ in \eqref{eq:sys_dynamics} is ergodic such that for all $i,j\in\Theta$, the mode transition from $i$ to $j$ within a finite number of steps has a non-zero probability.
\end{assumption}

\begin{assumption}
\label{assum4}
It is assumed that the system $\gss$ is mean-square stabilizable.
\end{assumption}

Then, the following theorem from \cite[Theorem 4.5]{costa2006discrete} provides the model-based optimal controller for the MJLS $\gss$ through solution of some coupled algebraic Riccati equations (CARE).

\begin{theorem}[\citealt{costa2006discrete}]
\label{th:1}
If the system $\gss$ in \eqref{eq:sys_dynamics} satisfies the assumptions~\ref{assum2}, \ref{assum3}, and \ref{assum4}, then, the optimal value of the cost function \eqref{eq:cost} is achieved by applying the control policy given by
\begin{equation}
u^*_k = - K_{\theta_k} x_k. \label{eq:control_policy}
\end{equation}
with the control gains for $i\in\Theta$ defined as
\begin{equation}
K_i = \left( R_i + B_i^T \mathcal{E}_i(P) B_i \right)^{-1} B_i^T \mathcal{E}_i(P) A_i \label{eq:control_gain}
\end{equation}
where the \( N \)-matrix \( P=(P_1,\ldots,P_N)\in\mathbb{R}^{n\times n\times N} \) has positive definite elements that satisfy 
\begin{equation}
\begin{split}
P_i = & Q_i + A_i^T \mathcal{E}_i(P) A_i \\
 -& A_i^T \mathcal{E}_i(P) B_i \left( R_i + B_i^T \mathcal{E}_i(P) B_i \right)^{-1} B_i^T \mathcal{E}_i(P) A_i 
\end{split} 
\label{eq:riccati_equations}
\end{equation}
for $i\in\Theta$ with the operator \( \mathcal{E}_i(P) \) defined as 
\begin{equation}
\mathcal{E}_i(P) = \mathbb{E}\left[  P_{\theta_{k+1}}  \mid \theta_k = i \right] = \sum_{j=1}^{N} p_{ij} P_j. \label{eq:operator_epsilon}
\end{equation}
Moreover, the optimal cost function denoted by $\jcost^*$ is
\begin{equation}
\jcost^*(x_0, \theta_0) = x_0^T P_{\theta_0} x_0. \label{eq:opt_cost}
\end{equation}
\end{theorem}

\subsection{Problem statement}
The controller gains in \eqref{eq:control_gain} are computed based on the information of the MJLS model in \eqref{eq:sys_dynamics}. 
The objective of this work is to compute the same set of controller gains only based on the measured system variables data and without need to use equation \eqref{eq:control_gain}.
In other words, no knowledge of the coefficients in \eqref{eq:sys_dynamics} or solution of the CARE in \eqref{eq:riccati_equations} are required.
The following additional assumptions are made.

\begin{assumption}
\label{assum5}
It is assumed that the information set $\mathcal{I}_k=\{I_0,I_1,\cdots,I_k\}$ on the system state variables is available to the controller at every time step $k\in\mathbb{Z}_0^+$.
\end{assumption}
\begin{assumption}
\label{assum1}
It is assumed that the transition probabilities in \eqref{eq:tpr} for the Markovian mode of $\gss$ in \eqref{eq:sys_dynamics} are known and available.
\end{assumption}
It is noteworthy that Assumption~\ref{assum1} is not restrictive, because the transition probabilities can be easily estimated empirically considering that the Markovian mode is measurable according to Assumption~\ref{assum5}.

\section{Value iteration for MJLS}\label{section:ValueIterationAlgorithm}
In this section, we employ the value iteration method to determine the optimal controller gain without the need to solve the CARE in \eqref{eq:riccati_equations}. Subsequently, we prove that the value function converges to the optimal value function along with the algorithm iterations. 

\subsection{Value iteration algorithm}
In this subsection, the Bellman principle of optimality is employed to introduce the value iteration algorithm for MJLS. The proposed method obtains the optimal control policies iteratively.
As the first step, the quadratic value function for the MJLS with a given controller is written as
\begin{equation}
\jcost(x_k, \theta_k) = \mathbb{E}\left[ \sum_{i=k}^{\infty}  x_i^T Q_{\theta_i} x_i + u_i^T R_{\theta_i} u_i  \mid I_k \right].\label{eq:quadratic_value_function}
\end{equation}

The cost function \eqref{eq:cost} is the above value function for $k=0$.
According to Equation~\eqref{eq:opt_cost} in Theorem~\ref{th:1}, the optimal value function $\jcost^*$ which is achieved by using the control gains \eqref{eq:control_gain} also takes the quadratic form
\begin{equation}
\jcost^*(x_k, \theta_k) = x_k^T P_{\theta_k} x_k. \label{eq:opt_value}
\end{equation}

Using the iterated conditional mathematical expectation \citep{papoulis2002probability}, the value function (\ref{eq:quadratic_value_function}) can be decomposed as
\begin{equation}
\begin{split}
& \jcost\left(x_k, \theta_k\right) =\mathbb{E}\bigg[x_k^T Q_{\theta_k} x_k + u_k^T R_{\theta_k} u_k +\\
& \qquad \mathbb{E}\Big[\sum_{i=k+1}^{\infty} x_i^T Q_{\theta_i} x_i + u_i^T R_{\theta_i} u_i \mid I_{k+1}\Big] \mid I_k\bigg] \text{.}
\end{split}
\end{equation}

By the definition of the value function (\ref{eq:quadratic_value_function}), the above equation can be expressed as
\begin{equation}
\begin{split}
\jcost(x_k,\theta_k)=x_k^T Q_{\theta_k} x_k&+u_k^T R_{\theta_k} u_k\\[3pt]
&+\mathbb{E}[ \jcost(x_{k+1},\theta_{k+1})| I_k]. \label{eq:Bellman}
\end{split}
\end{equation}

The above equation is recognized as the Bellman equation for system $\gss$. Based on the Bellman principle of optimality, we have 
\begin{equation}
\begin{aligned}
\jcost^* (x_k, \theta_k) = &\min_{u_k} \Big[ x_k^T Q_{\theta_k} x_k + u_k^T R_{\theta_k} u_k \\
&\quad\quad +\mathbb{E}\big[ J^* (x_{k+1}, \theta_{k+1}) \mid I_k \big] \Big]. \label{eq:HJB}
\end{aligned}
\end{equation}

This equation is a representation of the Hamiltonian-Jacobi-Bellman (HJB) equation for MJLS. According to the above equation, the optimal control policy is calculated as 
\begin{equation}
\begin{aligned}
u^*_k = \text{arg}\min_u \big[&x_k^T Q_{\theta_k} x_k + u_k^T R_{\theta_k} u_k \\
& + \mathbb{E}[\jcost^*(x_{k+1},\theta_{k+1}) | I_k ]\,\big]. 
\end{aligned}
\end{equation}

The value iteration algorithm computes sequences of control laws and value functions according to 
\begin{equation}
\begin{split}
u^j(x_k,\theta_k)=\text{arg}\min_{u_k}\big[&x_k^T Q_{\theta_k} x_k+u_k^T R_{\theta_k} u_k\\
&+\mathbb{E}[\jcost^j(x_{k+1},\theta_{k+1}) | I_k ]\,\big] \label{eq:control_policy_VI_algo}
\end{split}
\end{equation}
\begin{equation}
\label{eq:value_function}
\begin{split}
\jcost^{j+1}(x_k,\theta_k) =& x_k^T Q_{\theta_k} x_k + u^{j^T}(x_k, \theta_k) R_{\theta_k} u^j(x_k, \theta_k) \\
& +\mathbb{E}[\jcost^j(x_{k+1}, \theta_{k+1}) \mid I_k] 
\end{split}
\end{equation}
in which $j\in\mathbb{Z}_0^+$ is the iteration index and the initialization for $i=0$ is made as $\jcost^0(x_k,\theta_k)=0$.
The two equations \eqref{eq:control_policy_VI_algo} and \eqref{eq:value_function} can be combined as
\begin{equation}
\begin{split}
\jcost^{j+1}(x_k,\theta_k) = x_k^T Q_{\theta_k}& x_k + \min_{u_k} ~\big[~ u_k^T R_{\theta_k} u_k \\
& +\mathbb{E}[\jcost^j(x_{k+1}, \theta_{k+1}) \mid I_k] \,\big]
\label{eq:vi_combined}
\end{split}
\end{equation}

\subsection{Convergence proof of value iteration algorithm}
In this subsection, it is proved that the sequence of value functions given by \eqref{eq:value_function} converges to the optimal value function, i.e., $\jcost^j(x_k, \theta_k) \rightarrow \jcost^*(x_k, \theta_k)$, as $i \rightarrow \infty$. Let $\mu^j(x_k, \theta_k)$ be an arbitrary sequence of control policies, and a sequence $\Lambda^j(x_k, \theta_k)$, with the initial condition $\Lambda^0(x_k, \theta_k) = 0$, be recursively defined as
\begin{equation}\label{eq:lambda}
\begin{aligned}
\Lambda^{j+1}(x_k, \theta_k) &= x_k^T Q_{\theta_k} x_k + \mu^{j^T}(x_k, \theta_k) R_{\theta_k} \mu^j(x_k, \theta_k)\\
&\quad+ \mathbb{E}[\Lambda^j(x_{k+1}, \theta_{k+1}) \mid I_k].
\end{aligned} \raisetag{1.5\baselineskip}
\end{equation}

The following lemmas are necessary to prove the convergence of the value iteration algorithm for MJLS.

\begin{lemma}
\label{lemma1}
Let the value functions $\jcost^j(x_k, \theta_k)$ and $\Lambda^j(x_k, \theta_k)$ for the $j$-th iteration step be updated as in \eqref{eq:value_function} and \eqref{eq:lambda} with $\jcost^0(x_0, \theta_0) = \Lambda^0(x_0, \theta_0) = 0$. Then, the following inequality holds for all $j\in\mathbb{Z}_0^+$.
\begin{equation}\label{eq:22}
\jcost^j(x_k, \theta_k) \leq \Lambda^j(x_k, \theta_k).
\end{equation}
\end{lemma}
\begin{proof}
The proof begins by setting the initial conditions $\jcost^0(x_k, \theta_k) = \Lambda^0(x_k, \theta_k) = 0$. This establishes that the inequality \eqref{eq:22} holds at the initial step. Now, it is assumed that the inequality holds for some arbitrary step $j$, i.e., $\jcost^j(x_k, \theta_k) \leq \Lambda^j(x_k, \theta_k)$. It is necessary to show that this implies the inequality also holds for the next iteration step $i+1$. Based on the definitions of value functions, $\jcost^{j+1}(x_k, \theta_k)$ is obtained by minimizing the right-hand side of \eqref{eq:value_function} with respect to the optimal control policy, while $\Lambda^{j+1}(x_k, \theta_k)$ is the result of an arbitrary control policy. Therefore, it is concluded that
\begin{equation}\label{eq:23}
\jcost^{j+1}(x_k, \theta_k) \leq \Lambda^{j+1}(x_k, \theta_k).
\end{equation}
Hence, mathematical induction ensures that the inequality \eqref{eq:22} holds for all iteration steps. The proof is complete. 
\end{proof}

\begin{lemma}
\label{lemma2}
Consider the sequence of value functions $\jcost^j(x_k, \theta_k)$ given by \eqref{eq:value_function}. If there exists a control policy $\mu(x_k, \theta_k)$ that stabilizes the system $\gss$ in the mean-square sense, then there exists an upper bound $Y(x_k, \theta_k)$ such that $0 \leq \jcost^j(x_k, \theta_k) \leq \jcost^*(x_k, \theta_k) \leq Y(x_k, \theta_k)$ for all $j \in \mathbb{Z}_0^+$.
\end{lemma}

\begin{proof}
Let $\mu(x_k, \theta_k)$ be any stabilizing control policy, and set $\mu^j (x_k, \theta_k)$ in \eqref{eq:lambda} to $\mu(x_k, \theta_k)$ for $j \in \mathbb{Z}_0^+$. The difference between two consecutive iteration steps of $\Lambda^j (x_k, \theta_k)$ can be written as
\begin{equation*}
\begin{aligned}
&\Lambda^{j+1} (x_k, \theta_k) - \Lambda^j (x_k, \theta_k) \\
&= \mathbb{E}[\Lambda^j (x_{k+1}, \theta_{k+1}) \mid I_k] - \mathbb{E}[\Lambda^{j-1} (x_{k+1}, \theta_{k+1}) \mid I_k] \\
&= \mathbb{E}[\Lambda^{j-1} (x_{k+2}, \theta_{k+2}) \mid I_k] - \mathbb{E}[\Lambda^{j-2} (x_{k+2}, \theta_{k+2}) \mid I_k] \\
&= \mathbb{E}[\Lambda^{j-2} (x_{k+3}, \theta_{k+3}) \mid I_k] - \mathbb{E}[\Lambda^{j-3} (x_{k+3}, \theta_{k+3}) \mid I_k] \\
&\quad \vdots \\
&= \mathbb{E}[\Lambda^1 (x_{k+i}, \theta_{k+i}) \mid I_k] - \mathbb{E}[\Lambda^0 (x_{k+i}, \theta_{k+i}) \mid I_k].
\end{aligned}
\end{equation*}

Considering that $\Lambda^0$ is zero, the result is
\begin{equation*}
\Lambda^{j+1} (x_k, \theta_k) - \Lambda^j (x_k, \theta_k) = \mathbb{E}[\Lambda^1 (x_{k+i}, \theta_{k+i}) \mid I_k].
\end{equation*}

Summing up the above equation, we obtain
\begin{align}
\label{eq:lambdasum}
\Lambda^{j+1} (x_k, \theta_k) &=\sum_{n=0}^i \mathbb{E}[\Lambda^1 (x_{k+n}, \theta_{k+n}) | I_k ].
\end{align}

By the definition \eqref{eq:lambda}, we have 
\begin{equation} 
\Lambda^1(x,\theta) = x^T Q_\theta x + \mu^T(x,\theta) R_\theta \mu(x,\theta).
\end{equation} 

By defining
\begin{align} 
\label{eq:defy}
Y(x_k,& \theta_k) = \sum_{i=0}^\infty \mathbb{E}[
x_{k+i}^T Q_{\theta_{k+i}} x_{k+i} + \nonumber\\
&\mu^T(x_{k+i},\theta_{k+i}) R_{\theta_{k+i}} \mu(x_{k+i},\theta_{k+i})
|I_k ] 
\end{align}
one can deduce from \eqref{eq:lambdasum} that
\begin{equation} 
\Lambda^{j} (x_k, \theta_k)\leq Y(x_k, \theta_k) \qquad \forall j\in\mathbb{Z}_0^+
\end{equation}

Since the control policy \( \mu(x_k, \theta_k) \) stabilizes the system $\gss$ it is also concluded that
\begin{equation} 
\label{eq:ybnd}
Y(x_k, \theta_k) < \infty. 
\end{equation}

Using Lemma~\ref{lemma1}, we have
\begin{equation}
\label{eq:jley}
\jcost^{j} (x_k, \theta_k) \leq Y(x_k, \theta_k).
\end{equation}

According to \eqref{eq:defy} and \eqref{eq:quadratic_value_function}, we have $Y(x_k,\theta_k)=\jcost(x_k,\theta_k)$ if we set $u_k = \mu(x_k, \theta_k)$. 
Hence, in the special case $\mu(x_k, \theta_k) = u^*_k$ the value of $Y$ is minimized and becomes equal to $\jcost^*$ which implies 
\begin{equation}
\label{eq:jsley}
\jcost^*(x_k,\theta_k) \le Y(x_k,\theta_k).
\end{equation}

Since $Y$ is equal to $\jcost^*$ in the special case $\mu(x_k, \theta_k) = u^*_k$, the inequality \eqref{eq:jley} implies 
\begin{equation}
\label{eq:jlejs}
\jcost^{j} (x_k, \theta_k) \leq \jcost^*(x_k,\theta_k).
\end{equation}

It is concluded recursively from \eqref{eq:value_function} that $\jcost^{j}$ is non-negative for every $j\in\mathbb{Z}_0^+$ which together with \eqref{eq:ybnd}, \eqref{eq:jsley} and \eqref{eq:jlejs} completes the proof.
\end{proof}

\begin{lemma}
\label{lemma3}
Consider the sequence of value functions $\jcost^j (x_k, \theta_k)$ in \eqref{eq:value_function} with the control policies in \eqref{eq:control_policy_VI_algo} and $\jcost^0 (x_k, \theta_k) = 0$. Then, the sequence of value functions is non-decreasing, such that $\jcost^{j+1} (x_k, \theta_k) \geq \jcost^j (x_k, \theta_k) $ for all $j \in \mathbb{Z}_0^+$.
\end{lemma}

\begin{proof}
We attepmt to show by mathematical induction that $\Lambda^j (x_k, \theta_k) \leq \jcost^{j+1} (x_{k+1}, \theta_{k+1})$ for all $j\in\mathbb{Z}_0^+$. The induction starts by considering $\Lambda^0 (x_k, \theta_k) = \jcost^0 (x_k, \theta_k) = 0$ and using \eqref{eq:value_function} to write
\begin{equation}
\begin{split}
&\jcost^1 (x_k, \theta_k) - \Lambda^0 (x_k, \theta_k) \\
&= x_k^T Q_{\theta_k} x_k + u^{j^T} (x_k, \theta_k) R_{\theta_k} u^j (x_k, \theta_k) \geq 0
\end{split}
\end{equation}
which implies $\jcost^1 (x_k, \theta_k) \geq \Lambda^0 (x_k, \theta_k) $. 
Now, assume that $\jcost^j (x_k, \theta_k) \geq \Lambda^{j-1} (x_k, \theta_k)$ for some $i>0$. By setting  $\mu^{j-1} (x_k, \theta_k) =  u^{j} (x_k, \theta_k)$, and subtracting the lagged version of \eqref{eq:lambda} (in which $j$ is replaced by $j-1$) from \eqref{eq:value_function}, we have
\begin{equation}
\begin{split}
&\jcost^{j+1} (x_k, \theta_k) - \Lambda^j (x_k, \theta_k) \\
&= \mathbb{E}[\jcost^j (x_{k+1}, \theta_{k+1}) - \Lambda^{j-1} (x_{k+1}, \theta_{k+1}) |I_k]\geq 0 \label{eq:32}
\end{split}
\end{equation}
which proves that $ \jcost^{j+1} (x_{k+1}, \theta_{k+1}) \geq \Lambda^j (x_k, \theta_k)$. 
This inequality together with \eqref{eq:22} from Lemma~\ref{lemma1} implies $\jcost^{j+1} (x_k, \theta_k) \geq \jcost^j (x_k, \theta_k) $ for every $j\in\mathbb{Z}_0^+$ which proves the result.
\end{proof}

\begin{remark}
\label{rem:1}
Considering that the sequence of value functions defined in \eqref{eq:value_function} is non-decreasing and has an upper bound, it is concluded that the following limit exists 
\begin{equation}
\label{eq:jlim}
\lim_{i \to \infty} \jcost^j (x_k, \theta_k) = \jcost^\infty (x_k, \theta_k).
\end{equation} 
\end{remark}

\begin{theorem} 
\label{th:2}
Let $\jcost^j (x_k, \theta_k)$ denote the sequence of value functions given by \eqref{eq:value_function}. If the system $\gss$ in \eqref{eq:sys_dynamics} satisfies the assumptions~\ref{assum2}, \ref{assum3}, and \ref{assum4}, then the limit $\jcost^\infty (x_k, \theta_k)$ in \eqref{eq:jlim} is equal to the optimal value function $\jcost^* (x_k, \theta_k)$ and satisfies the HJB equation \eqref{eq:HJB} as well as the Bellman equation \eqref{eq:Bellman} for the optimal control policy $u_k = u^* (x_k, \theta_k)$.
\end{theorem}

\begin{proof}
According to Lemma~\ref{lemma2} we have
\begin{equation}
\label{eq:jilejs}
\jcost^\infty (x_k, \theta_k) \le \jcost^* (x_k, \theta_k).
\end{equation}
For $j \to \infty$, the equation \eqref{eq:vi_combined} becomes
\begin{align}
\jcost^\infty (x_k, \theta_k) = & ~x_k^T Q_{\theta_k} x_k + u^\infty(x_k, \theta_k)^T R_{\theta_k} u^\infty(x_k, \theta_k) \nonumber\\[4pt]
& + \mathbb{E}[\jcost^\infty (x_{k+1}, \theta_{k+1}) | I_k].
\end{align}
Substituting $k$ in the above equation with $i$ for $i\ge k$, taking conditional expectation given $I_k$, summing over $i$, and simplifying the result, we arrive at
\begin{align}
\jcost^\infty (x_k, \theta_k) &=  \mathbb{E}\bigg[\sum_{i=k}^\infty [x_i^T Q_{\theta_i} x_i + \nonumber\\[-6pt]
& u^\infty(x_i, \theta_i)^T R_{\theta_i} u^\infty(x_i, \theta_i)] ~|~ I_k\bigg].
\end{align}

According to the above equation, $\jcost^\infty$ is $\jcost$ in \eqref{eq:quadratic_value_function} if we set $u_k=u^\infty(x_k, \theta_k)$. Hence, $\jcost^\infty$ connot be less than the optimal value of $\jcost$ and we have
\begin{equation}
\label{eq:jsleji}
\jcost^\infty (x_k, \theta_k) \ge \jcost^* (x_k, \theta_k).
\end{equation}

The two inequalities \eqref{eq:jilejs} and \eqref{eq:jsleji} together imply that $\jcost^\infty=\jcost^*$ and $u^\infty=u^*$. As a result, \eqref{eq:HJB} is satisfied by $\jcost^*=\jcost^\infty$ and \eqref{eq:Bellman} is satisfied by $\jcost=\jcost^\infty$ with $u_k=u^\infty(x_k,\theta_k)$ which completes the proof.
\end{proof}

\begin{corollary}
\label{col:1}
Let $\jcost^j (x_k, \theta_k)$ denote the sequence of value functions given by \eqref{eq:value_function}. If the system $\gss$ in \eqref{eq:sys_dynamics} satisfies the assumptions~\ref{assum2}, \ref{assum3}, and \ref{assum4}, then the positive definite matrices $P^j_i$ exist for $i\in\Theta$ and $j\in\mathbb{Z}_0^+$ such that 
\begin{subequations}
\begin{align}
\jcost^j (x, i) &= x^T P^j_i x \label{eq:pjq}\\[4pt]
P_{i}^{j+1}&=Q_{i} + A_{i}^T \mathcal{E}_{i}(P^j) A_{i} - A_{i}^T \mathcal{E}_{i}(P^j) B_{i}\times\nonumber\\
&\qquad (R_{i} + B_{i}^T \mathcal{E}_{i}(P^j) B_{i})^{-1} B_{i}^T \mathcal{E}_{i}(P^j) A_{i}\label{eq:pjr}\\[4pt]
\lim_{j \to \infty} P^j_{i} &= P_{i} \label{eq:42}
\end{align}
\end{subequations}
for all $i\in \Theta$ where $P_{i}$ are the solutions of the CARE in \eqref{eq:riccati_equations}.
\end{corollary}

\begin{proof}[Sketch of the proof]
The proof of equations \eqref{eq:pjq} and \eqref{eq:pjr} can be established by induction using equation \eqref{eq:vi_combined} given that $\jcost^0=0$, which implies $P_{i}^0=0$. Consequently, \eqref{eq:42} is resulted from \eqref{eq:jlim} and Theorem~\ref{th:2} which necessitates $\jcost^\infty=\jcost^*$.
\end{proof}

Theorem~\ref{th:2} ensures the convergence of the sequence of value functions $\jcost^j$ for $j\in\mathbb{Z}_0^+$ to the optimal value function $\jcost^*$ throughout the algorithm iterations. 
Also, Corollary~\ref{col:1} ensures the convergence of the corresponding sequence of kernel matrices to the solution of the CARE in \eqref{eq:riccati_equations}. 
These result will be leveraged to the convergence proof of the Q-learning approach outlined in Section~\ref{section 4.4}. 

Value iteration is a computationally efficient method for solving the Riccati equations and finding the optimal controller for MJLS. However, this algorithm relies on having the knowledge of the system dynamics. To address this limitation, Q-learning is proposed as a comprehensive solution for model-free optimal control of MJLS in the next section.

\section{Q-learning}\label{section:Q-Learning}
In this section, we first construct the Q-function for model-free optimal control of MJLS. Then, we develop a mechanism for estimating the kernel matrix using the LS method. 
Unlike the model-based LQR and value iteration algorithm, the resulting Q-learning based MJLS optimal control design procedure depends solely on the input-state information of the system.
Finally, we prove that the proposed model-free controller converges to the model-based controller and we show that the estimator is unbiased against excitation noise. 

\subsection{Q-function for optimal control}
The stochastic Q-function represents the value associated with taking an action $u_k$ in a specific state $x_k$, mode $\theta_k$, and following the optimal policy after the time step $k$. It is defined as follows
\begin{equation}
\begin{split}
\mathbb{Q}(x_k,u_k,\theta_k) &= x_k^T Q_{\theta_k} x_k + u_k^T R_{\theta_k} u_k \\
&\qquad+\mathbb{E}\left[\ \jcost^*(x_{k+1},\theta_{k+1}) \ \, \middle| \, I_k\right].\label{eq:Q_function}
\end{split}
\end{equation}

The stochastic Q-function proposed above is closely related to the value function for system $\gss$. This relationship can be formulated by representing the HJB equation \eqref{eq:HJB} in terms of the Q-function as
\begin{equation}
\jcost^*(x_k,\theta_k) = \min_{u_k} \mathbb{Q}(x_k,u_k,\theta_k). \label{eq:bellman2}
\end{equation}

Subsequently, we have
\begin{equation}
u_k^* = \arg\min_{u_k} ~\mathbb{Q}(x_k,u_k,\theta_k). \label{eq:optimal_control}
\end{equation}

It is usual and useful to define a stage cost function such as below to write down the calculations more compactly.  
\begin{equation}
r(x,u,\theta)=x^T Q_\theta x + u^T R_\theta u .
\end{equation}

Substitution of \eqref{eq:opt_value} and \eqref{eq:sys_dynamics} into \eqref{eq:Q_function} while using the above notation results in
\begin{equation}
\begin{aligned}
\mathbb{Q}(x_k,u_k,\theta_k) &= r(x_k,u_k,\theta_k)+\left( A_{\theta_k} x_k + B_{\theta_k} u_k \right)^T \\
& \times \mathbb{E}\left[  P_{\theta_{k+1}} \, \middle| \, I_k \right] \left( A_{\theta_k} x_k + B_{\theta_k} u_k \right).
\end{aligned}
\end{equation}

Expanding the products and using \eqref{eq:operator_epsilon}, we have
\begin{equation}
\begin{aligned}
\mathbb{Q}(&x_k,u_k,\theta_k) = r(x_k,u_k,\theta_k) \\
& + x_k^T A_{\theta_k}^T \mathcal{E}_{\theta_k}(P) A_{\theta_k} x_k + u_k^T B_{\theta_k}^T \mathcal{E}_{\theta_k}(P) B_{\theta_k} u_k\\
&+ x_k^T A_{\theta_k}^T \mathcal{E}_{\theta_k}(P) B_{\theta_k} u_k+ u_k^T B_{\theta_k}^T \mathcal{E}_{\theta_k}(P) A_{\theta_k} x_k . 
\end{aligned}
\end{equation}

The Q-function can be written in terms of the augmented vector $z_k = [x_k^T ~~ u_k^T]^T$ as
\begin{equation}
\mathbb{Q}(x_k,u_k,\theta_k) = z_k^T H_{\theta_k} z_k. \label{eq:Q-function_quadratic}
\end{equation}
where $H_{\theta_k} = H_{\theta_k}^T \in \mathbb{R}^{l \times l}$, $l=n+m$ is denoted as the kernel matrix of the Q-function. Based on $(45)$ and $(46)$, the kernel matrix is decomposed as
\begin{subequations}
\label{eq:kernela}
\begin{align}
H_{\theta_k} =& \begin{bmatrix}
H_{\theta_k}^{xx} & H_{\theta_k}^{xu} \\
H_{\theta_k}^{ux} & H_{\theta_k}^{uu}
\end{bmatrix},\label{eq:kernel}\\[6pt]
H_{\theta_k}^{xx} &= Q_{\theta_k} + A_{\theta_k}^T \mathcal{E}_{\theta_k}(P) A_{\theta_k}, \\[3pt]
H_{\theta_k}^{xu} &= A_{\theta_k}^T \mathcal{E}_{\theta_k}(P) B_{\theta_k}, \\[3pt]
H_{\theta_k}^{ux} &= B_{\theta_k}^T \mathcal{E}_{\theta_k}(P) A_{\theta_k}, \\[3pt]
H_{\theta_k}^{uu} &= R_{\theta_k} + B_{\theta_k}^T \mathcal{E}_{\theta_k}(P) B_{\theta_k}.
\end{align}
\end{subequations}

It is straightforward to use the equations \eqref{eq:kernela} and \eqref{eq:control_gain} to write \eqref{eq:riccati_equations} as
\begin{equation}
\label{eq:ph}
P_{i} = \begin{bmatrix} I & K_{i}^T \end{bmatrix} H_{i} \begin{bmatrix} I \\ K_{i} \end{bmatrix}.
\end{equation}

We can now formulate the optimal controller using the proposed Q-function. 
According to \eqref{eq:optimal_control}, the optimal control policy \( u_k^* \) is determined from the following equation.
\begin{equation}
\frac{\partial}{\partial u} \mathbb{Q}^*(x_k,u^*_k,\theta_k) = 0 \label{eq:optimal_policy}
\end{equation}
which is solved based on \eqref{eq:Q-function_quadratic} and \eqref{eq:kernela} as
\begin{equation}
u_k^* = - (H_{\theta_k}^{uu})^{-1} H_{\theta_k}^{ux} x_k. \label{eq:optimal_policy_solution}
\end{equation}

In contrast to the well-established conventional approaches, if the information of Q-function is available, the the proposed method can be used  without the need for any prior knowledge of the system dynamics. Furthermore, the optimal control policy is achieved without the necessity to solve nonlinear CARE \eqref{eq:riccati_equations} in an offline manner.

The primary challenge within the proposed method centers on obtaining the kernel matrix \(H_{\theta_k}\) through the interaction of an intelligent agent with the unknown system $\gss$. In the next subsection, the policy evaluation and improvement steps are developed to learn the kernel matrix for the system $\gss$.

\subsection{Formulation of Q-learning}
In this subsection, a learning-based method is developed to achieve the optimal control policy using the input and state data for the system $\gss$. 
The RL-based optimal control method for MJLS is based on the equation \eqref{eq:Q_function}
which is used to estimate \(H_{i}\) for $i \in \Theta$ in policy evaluation phase using the collected system data.
Then, in policy improvement phase, the estimated kernel matrix is used to calculate the controller gains based on \eqref{eq:optimal_policy_solution} in a model-free framework. The iteration between policy evaluation and improvement continues until the convergence criterion is met. 

Denoting the iteration index by $j$, the policy evaluation phase of the Q-learning algorithm based on value iteration is implemented by writing \eqref{eq:Q_function} as
\begin{equation}
\begin{split}
\mathbb{Q}^{j+1}(x_k,u_k,\theta_k) &= r(x_k,u_k,\theta_k)\\
&+ \mathbb{E}[ J^{*j}(x_{k+1},\theta_{k+1})|I_k]. \label{eq:qfuniter}
\end{split}
\end{equation}
where $\mathbb{Q}^j$ and $\jcost^{*j}$ are the estimates of $\mathbb{Q}$ and $\jcost^{*}$ obtained at the $j$-th iteration respectively.

Using \eqref{eq:Q-function_quadratic}, \eqref{eq:opt_value}, the notation \eqref{eq:operator_epsilon}, and Assumption~\ref{assum2}, the equation \eqref{eq:qfuniter} can be written in terms of the kernel matrices as
\begin{equation}
\begin{split}
z_k^T H^{j+1}_{\theta_k} z_k &= r(x_k,u_k,\theta_k)
+ x_{k+1}^T \mathcal{E}_{\theta_k}(P^{j}) x_{k+1}.
\label{eq:qfuniterm}   
\end{split}
\end{equation}
where $H^j_i$ and $P^j_i$ are respectively the values of $H_i$ and $P_i$ obtained at the $j$-th iteration for $i\in\Theta$.

\subsection{Implementation of Algorithm}
In this subsection, we develop LS method to estimate the unknown kernel matrix using sampled data generated by MJLS. 
For this purpose, it is essential to separate the unknown kernel matrix at iteration $j+1$ in the left hand side of \eqref{eq:qfuniterm} as
\begin{equation}
z_k^T H^{j+1}_{\theta_k} z_k = \bar{z}_k^T \bar{H}^{j+1}_{\theta_k} \label{eq:parameterization}
\end{equation}
where \(\bar{H}^j_{\theta_k}, \bar{z}_k \in \mathbb{R}^{\frac{l(l+1)}{2}}\) are defined in terms of the elements of $H^j_{\theta_k}$ and $z_k$ as
\begin{align}
\bar{H}^j_{\theta_k} &= [h^j_{11}, 2h^j_{12}, \dots, 2h^j_{1l}, \nonumber\\
&\qquad h^j_{22}, \dots, 2h^j_{2l}, \dots, h^j_{ll}]^T,\\
\bar{z}_k &= [z_{k,1}^2, z_{k,1} z_{k,2}, \dots, z_{k,1}z_{k,l}, \nonumber\\ &\qquad z_{k,2}^2, \dots, z_{k,2}z_{k,l}, \dots, z_{k,l}^2]^T. \label{eq:regression_vector}
\end{align}

Due to the symmetry of the kernel matrix \(H^j_{\theta_k}\), its off-diagonal elements are weighted by a factor of 2.
Substituting \eqref{eq:parameterization} into \eqref{eq:qfuniterm} the result is
\begin{equation}
\bar{z}_k^T \bar{H}^{j+1}_{\theta_k} = r(x_k,u_k,\theta_k) + x_{k+1}^T \mathcal{E}_{\theta_k}(P^{j}) x_{k+1}.
\label{eq:substitution}
\end{equation}

The dataset for the LS method is collected through the interaction of an intelligent agent with the unknown system $\gss$ for the policy evaluation phase at each iteration. 
The data is collected such that \(L\) data samples are available for each mode value in $\Theta$. 
To be more precise, consider the dataset $(u_k,x_k,\theta_k)$ for $k\in\{1,2,\cdots,\bar L\}$ in the $j$-th iteration, where the iteration index $j$ is not indicated for brevity. The index variables $r_{i,\ell}$ for $i\in\theta$ and $\ell\in\{1,2,\cdots,M_i\}$ with $M_i\ge L$ are defined such that 
\begin{itemize}
\setlength{\itemindent}{7mm}
\item $r_{i,\ell_1}<r_{i,\ell_2}$ if and only if $\ell_1<\ell_2$. 
\item $\theta_k=i$ if and only if $k\in\{r_{i,1},r_{i,2},\cdots,r_{i,M_i}\}$. 
\end{itemize}

The samples of dataset for the $j$-th iteration are used to construct 
\begin{subequations}
\label{eq:phi and sie}
\begin{align}
\phi_i^j &= \begin{bmatrix}
\bar{z}_{r_{i,1}} &
\bar{z}_{r_{i,2}} &
\cdots &
\bar{z}_{r_{i,L}}
\end{bmatrix}^T, \\
\psi_i^j &= \begin{bmatrix}
r(x_{r_{i,1}},u_{r_{i,1}},i) + x_{r_{i,1}+1}^T \mathcal{E}_{i}(P^j) x_{r_{i,1}+1}\\
r(x_{r_{i,2}},u_{r_{i,2}},i) + x_{r_{i,2}+1}^T \mathcal{E}_{i}(P^j) x_{r_{i,2}+1} \\
\vdots \\
r(x_{r_{i,L}},u_{r_{i,L}},i) + x_{r_{i,L}+1}^T \mathcal{E}_{i}(P^j) x_{r_{i,L}+1}
\end{bmatrix}.
\end{align}
\end{subequations}

Then, the equation \eqref{eq:substitution} for \(k \in \{1,\ldots,L\}\) can be transformed into the following stochastic linear regression model
\begin{equation}
\label{eq:regrm}
\phi^{j}_i \bar{H}^{j+1}_i = \psi^{j}_i, \quad \forall i\in \Theta.
\end{equation}

According to the above equation, the sequence of estimations for the kernel matrix \( \hat{H}_i^{j+1} \) are obtained as
\begin{equation}
\hat{H}_i^{j+1} = \arg \min_{\bar{H}} \| \psi_i^j - \phi_i^j \bar{H} \|, \quad \forall i \in \Theta
\end{equation}
which is solved as
\begin{equation}\label{eq:Hhat}
\hat{H}_i^{j+1} = (\phi_i^{j^T} \phi_i^j)^{-1} \phi_i^{j^T} \psi_i^j, \quad \forall i \in \Theta.
\end{equation}

\begin{remark}
\label{rem:2}
The kernel matrix is symmetric with \( \frac{l(l+1)}{2} \) independent entries. Consequently, a minimum of \( L \geq \frac{l(l+1)}{2} \) data samples are required for each Markovian mode in order to use \eqref{eq:Hhat}.
According to Assumption~\ref{assum3}, it is ensured that each mode $i\in\Theta$ is visited for $L$ times if the total length of data samples during the $j$-th iteration of learning process denoted by $\bar L$ is large enough \citep{sayedana2024strong}. 
To derive a rule of thumb, one can consider the steady state distribution $\pi_\infty$ for the mode $\theta_k$ to estimate the probability distribution of $\bar L$ as
\begin{align}
\mathbb{P}[\bar L=\beta] =  \sum_{\mathclap{\substack{k_i\ge L,\\ \scalebox{1}{$\Sigma$}_{i=1}^N k_i=\beta}}} C_\beta(k_1,k_2,\cdots,k_N) \prod_{\ell=1}^N \pi_{\infty,\ell}^{k_i}
\end{align}
where $C_\beta(k_1,k_2,\cdots,k_N)$ is combination of $k_1,k_2,\cdots,k_N$ from $\beta$ (the number of ways to select $N$ sets of $k_1,k_2,\cdots~k_N$ places from a total of $\beta$ places). 
The above probability distribution can be used for example to estimate the average or standard deviation of the required data length $\bar L$.
\end{remark}

To guarantee the invertibility of \( \phi_i^T \phi_i \) in \eqref{eq:Hhat}, the input sequence is required to fulfill a persistent excitation condition according to the following definition.

\begin{definition}[\citealt{petreczky2023notion}] 
\label{defin4}
The control input sequence applied to the system $\gss$ is said to meet the persistence excitation condition if there exists a constant \( \epsilon > 0 \) such that
\begin{equation}
\sum_{k=0}^{L} u_k u_k^T \geq \epsilon I.
\end{equation}
\end{definition}

One way to satisfy the persistence excitation condition is to add an excitation noise \( n_k \) to the input signal, and apply $\hat u_k$ defined below as the input instead of $u_k$. 
\begin{equation}
\hat{u}_k = u_k + n_k. \label{eq:noisy_input}
\end{equation}

Adding noise according to the above equation does not cause any bias in estimation of the kernel matrix (see Remark~\ref{rem:3} in the next subsection).

Summarizing the whole procedure, the resulting MJLS Q-learning algorithm is presented as Algorithm~\ref{alg:q-learning}.

\begin{algorithm}
\caption{Q-Learning Algorithm}
\label{alg:q-learning}
\begin{algorithmic}[1]
\Statex \textbf{Initialization:} $K_i^0$ and $\bar{H}_i^0 \geq 0$ for each $i \in \Theta$.
\Statex \textbf{Parameter:} error bound $\varepsilon$ for stop condition.
\State $e_K:=\infty$
\While{$e_K>\varepsilon$}
    \State Acquire a dataset $(\hat u_k,x_k,\theta_k)$ of system $\gss$ 
    \Statex \hspace{6mm} with $k\in\{0,1,\cdots,\bar L-1\}$ which contains at 
    \Statex \hspace{6mm} least $L$ samples for each mode value.
    \For{$i\in\Theta$}
    \State Construct the matrices in \eqref{eq:phi and sie}.
    \State Calculate $\hat{H}_i^{j+1}$ using \eqref{eq:Hhat}.
    \State Calculate $K_i^{j+1}$ using \eqref{eq:optimal_policy_solution} and \eqref{eq:kernel}.
    \EndFor
    \State {$e_K:=\max_{i\in\Theta} \|K_i^{j+1} - K_i^{j}\|$}
\EndWhile
\end{algorithmic}
\end{algorithm}

\subsection{Convergence of algorithm}\label{section 4.4}
In this subsection, we investigate the convergence of the sequence of model-free controller gains in Algorithm~\ref{alg:q-learning} to the model-based optimal controller gains defined in Theorem~\ref{th:1}.
Then, we ensure the immunity of the estimation mechanism against the bias phenomenon.

\begin{theorem} 
\label{th:4}
For the mean-square stabilizable system $\gss$, the sequence of model-free optimal controller gains generated by the Q-learning algorithm \( K_{i}^j\) for $i \in \Theta$ and \( j\in\mathbb{Z}_0^+\), converges to the model-based optimal controller gains in \eqref{eq:control_gain} such that 
\begin{equation}
K_{i} = \lim_{j \rightarrow \infty} K_{i}^j .
\end{equation}
\end{theorem}

\begin{proof}
Substituting $z_k$ in \eqref{eq:qfuniterm} with $[x_k^T ~~ u_k^T]^T$ and replacing $x_{k+1}$ from the system equation \eqref{eq:sys_dynamics} we have
\begin{equation}
\begin{split}
&\begin{bmatrix} x_k^T & u_k^T \end{bmatrix} H_{i}^{j+1} \begin{bmatrix} x_k \\ u_k \end{bmatrix} \\
&= \begin{bmatrix} x_k^T & u_k^T \end{bmatrix} \begin{bmatrix} Q_{i} & 0 \\ 0 & R_{i} \end{bmatrix} \begin{bmatrix} x_k \\ u_k \end{bmatrix} \\
&+ \begin{bmatrix} x_k^T & u_k^T \end{bmatrix} \begin{bmatrix} A_{i}^T & B_{i}^T \end{bmatrix} \mathcal{E}_{i} (P^j) \begin{bmatrix} A_{i} & B_{i} \end{bmatrix} \begin{bmatrix} x_k \\ u_k \end{bmatrix}.
\end{split}\label{eq:78}
\end{equation}

Since the above equation is valid for all values of \(x_k\) and \(u_k\), one can obtain
\begin{align}
H_{i}^{j+1} &= \begin{bmatrix}
Q_{i} & 0 \\
0 & R_{i}
\end{bmatrix}
+ \begin{bmatrix}
A_{i} & B_{i}
\end{bmatrix}^T \mathcal{E}_{i} (P^j) \begin{bmatrix}
A_{i} & B_{i}
\end{bmatrix} \nonumber\\
&= \begin{bmatrix}
Q_{i} + A_{i}^T \mathcal{E}_{i} (P^j) A_{i} & A_{i}^T \mathcal{E}_{i} (P^j) B_{i} \\
B_{i}^T \mathcal{E}_{i} (P^j) A_{i} & R_{i} + B_{i}^T \mathcal{E}_{i} (P^j) B_{i}
\end{bmatrix}.
\label{eq:79}
\end{align}

Using \eqref{eq:ph}, we have
\begin{align}
P_{i}^{j+1} &=\begin{bmatrix}
I & K_{i}^{{j+1}^T}
\end{bmatrix}H_{i}^{j+1}\begin{bmatrix}
I \\ K_{i}^{j+1}
\end{bmatrix}.
\label{eq:80}
\end{align}

Also, \eqref{eq:79} in combination with \eqref{eq:optimal_policy_solution} gives
\begin{equation}
\begin{split}\label{eq:81}
K_{i}^{j+1} &= (R_{i} + B_{i}^T \mathcal{E}_{i}(P^j) B_{i})^{-1} B_{i}^T \mathcal{E}_{i}(P^j) A_{i}.
\end{split}
\end{equation}

Substituting \eqref{eq:79} and \eqref{eq:81} into \eqref{eq:80}, results in
\begin{align}
&P_{i}^{j+1}=Q_{i} + A_{i}^T \mathcal{E}_{i}(P^j) A_{i} \nonumber\\
&- A_{i}^T \mathcal{E}_{i}(P^j) B_{i} (R_{i} + B_{i}^T \mathcal{E}_{i}(P^j) B_{i})^{-1} B_{i}^T \mathcal{E}_{i}(P^j) A_{i}.
\end{align}

The above equation is the same as \eqref{eq:pjr} in Corollary~\ref{col:1}, which necessitates $\lim_{j \to \infty} P_{i}^j=P_{i}$ according to \eqref{eq:42}. 
As a result \eqref{eq:81}, implies convergence of $K^j_i$ to $K_i$ in \eqref{eq:control_gain} which completes the proof.
\end{proof}

\begin{remark}
\label{rem:3}
The addition of noise according to \eqref{eq:noisy_input} does not affect the correctness or convergence of Algorithm~\ref{alg:q-learning}. The reason is that in the proof of Theorem~\ref{th:4}, $u_k$ in Equation \eqref{eq:78} can be easily replaced with $\hat{u}_k$ without affecting the derivation of \eqref{eq:79}. The rest of the proof is independent of the choice of input and follows as before.
\end{remark}

\begin{remark}
\label{rem:4}
In accordance with Theorem~\ref{th:4}, the model-free controller gains converge towards the conventional optimal controller gains \eqref{eq:control_gain}, which are derived through the solution of Riccati equations \eqref{eq:riccati_equations}. Consequently, the Theorem~\ref{th:1} guarantees the stability of the system \eqref{eq:sys_dynamics} controlled using the proposed algorithm in the mean-square sense of Definition~\ref{defin2}, without requiring prior knowledge of the system dynamics.
\end{remark}

\section{Simulation analysis}\label{section:Simulation}
In this section, the performance of the proposed model-free optimal controller for  discrete-time MJLS is evaluated under various operating conditions, including different excitation noise sequences. The aim is to assess how the proposed model-free controller utilizes raw data to regulate the state of an MJLS. Finally, the convergence of the model-free controller to the model-based controller gains, as well as the closed-loop stability of the system, is investigated. For this purpose, a discrete-time MJLS with two Markovian modes is adopted from \cite{wan2023self}, with system matrices  
\begin{align*}
A_1 &= \begin{bmatrix} -0.5 & 1 \\ 0.8 & 0.5 \end{bmatrix},\quad A_2 = \begin{bmatrix} 0.6 & -0.1 \\ 0.4 & -1 \end{bmatrix}\\
B_1 &= \begin{bmatrix} 1 \\ 2 \end{bmatrix}, \quad B_2 = \begin{bmatrix} 1 \\ 1 \end{bmatrix}    
\end{align*}
in which $A_1$ is an unstable matrix. The transition probability matrix is also considered to be 
\begin{equation*}
\Phi = \begin{bmatrix} 0.7 & ~~0.3 \\ 0.5 & ~~0.5 \end{bmatrix}.
\end{equation*}

The initial state is assumed to be \( x_0 = [1 ~~ 0]^T \) and the initial mode is considered as \( \theta_0 = 1 \). The simulation tests are performed for 200 time steps. The transitions of the Markovian mode for 50 samples are presented in Fig.~\ref{fig:1}.
\begin{figure}[H]
    \centering
    \includegraphics[width=\linewidth]{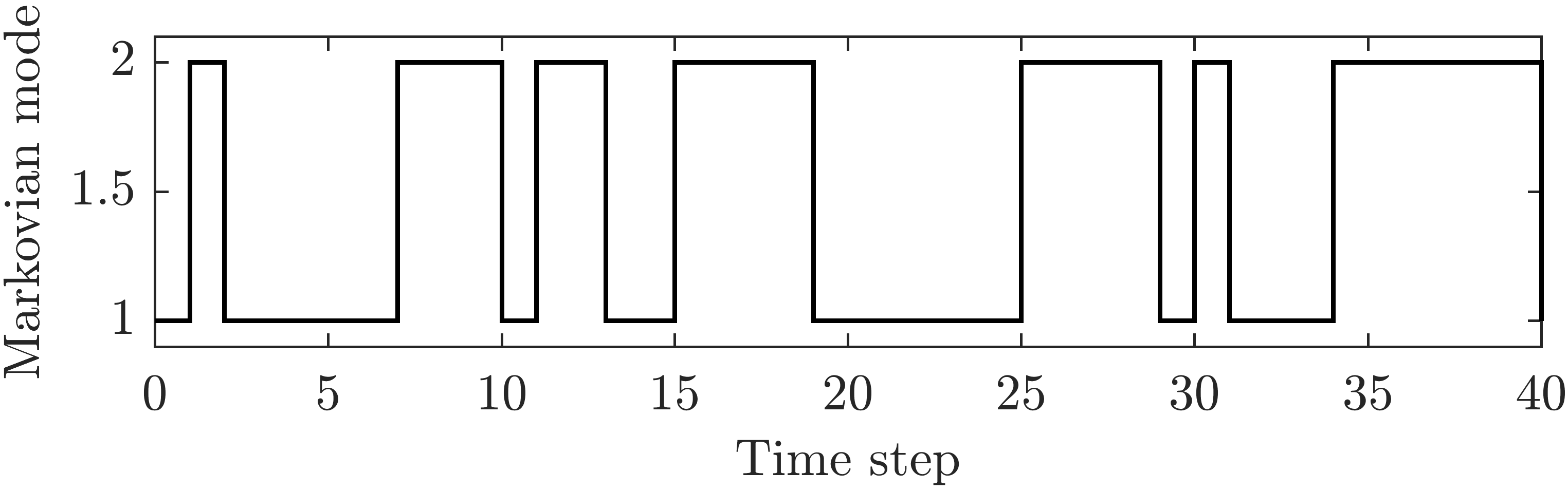}
    \caption{Transitions of the Markovian mode}
    \label{fig:1}
\end{figure}
The weighting matrices of Q-function are specified as \( Q_1 = Q_2 = 5I \) and \( R_1 = R_2 = 1 \). Based Theorem~\ref{th:1}, the model-based optimal controller gains are calculated using Corollary~\ref{col:1} as
\begin{align*}
K_1 &= \big[0.354 ~~~~ 0.273\big] \\
K_2 &= \big[0.460 ~~~ -0.551\big]
\end{align*}

In each iteration, $L=15$ data samples are collected to estimate the kernel matrix for each mode using the LS method during the policy evaluation phase. The stop condition parameter of Algorithm~\ref{alg:q-learning} is set to \( \epsilon = 0.001 \). The algorithm is initialized as 
\begin{align*}
\bar{H}_1^0 &= \bar{H}_2^0 = \big[ 0 ~~ 0 ~~ 0 ~~ 0 ~~ 0 ~~ 0  \big]^T \\
K_1^0 &= K_2^0 = \big[0 ~~~ 0\big] 
\end{align*}

The performance of the closed-loop system under the proposed model-free controller is illustrated in Fig.~\ref{fig:2}. In this simulation, the excitation noise is characterized as a Gaussian process with zero mean and standard deviation of 0.01. The system response is plotted for controller gains at iterations $0, 2, 5, 25$. According to this figure, the control performance quickly approaches to the optimal case as the iteration count increases. The convergence takes place in 25 iterations. The controller at final iteration optimally and effectively regulates the system states without prior knowledge of system dynamics or the need to solve the Riccati equations.
The controller gain elements are plooted in Fig.~\ref{fig:3} where the $j$-th element of controller gain for mode $i\in\Theta$ is denoted as $K_{i,j}$. This figure also shows the rapid convergence of the model-free controller gains to the model-based optimal controller gains as stated in Theorem~\ref{th:4}.

\begin{figure}
    \centering
    \includegraphics[width=\linewidth]{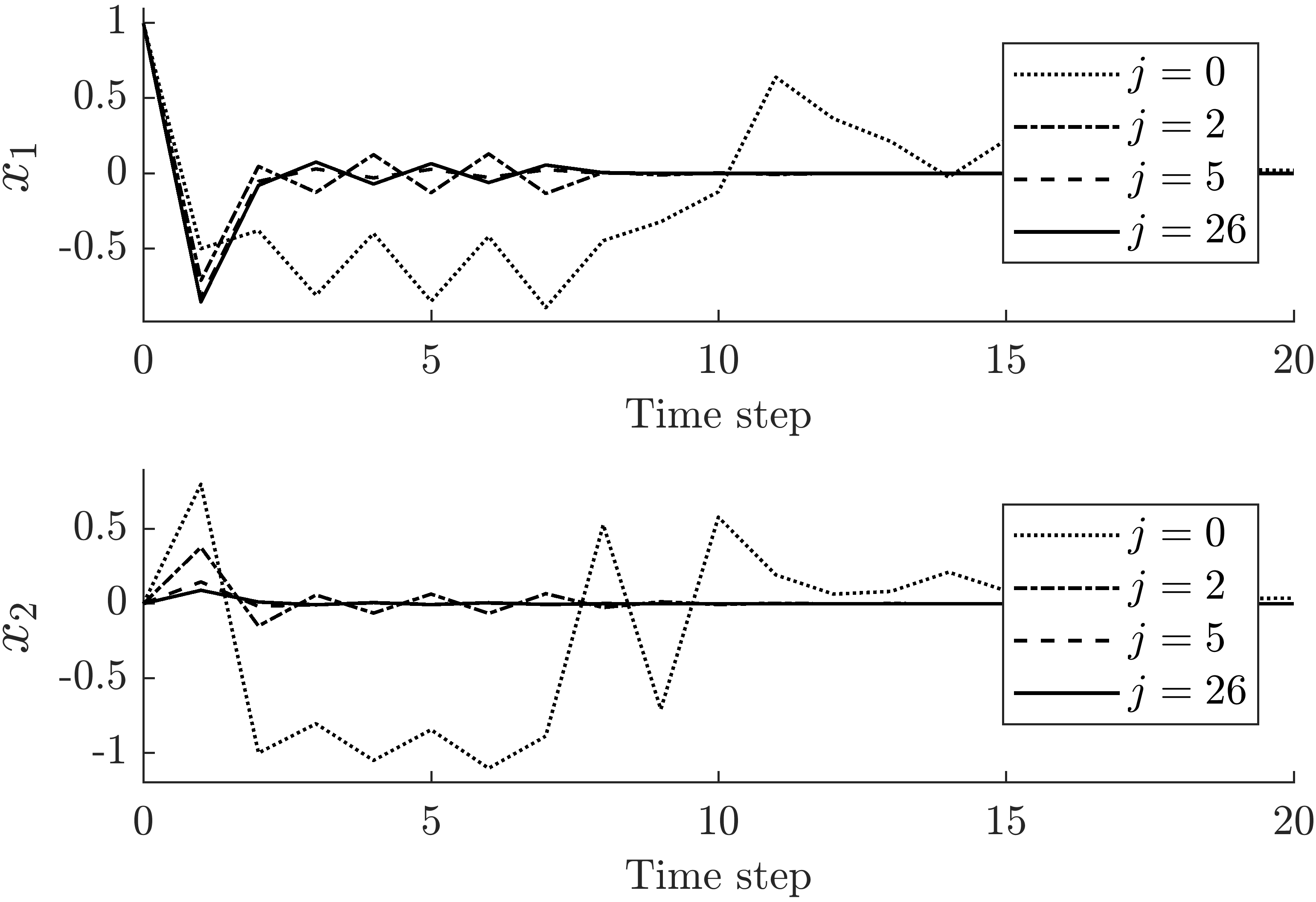}
    \caption{Model-free controller during iterations.}
    \label{fig:2}
\end{figure}

\begin{figure}
    \centering
    \includegraphics[width=\linewidth]{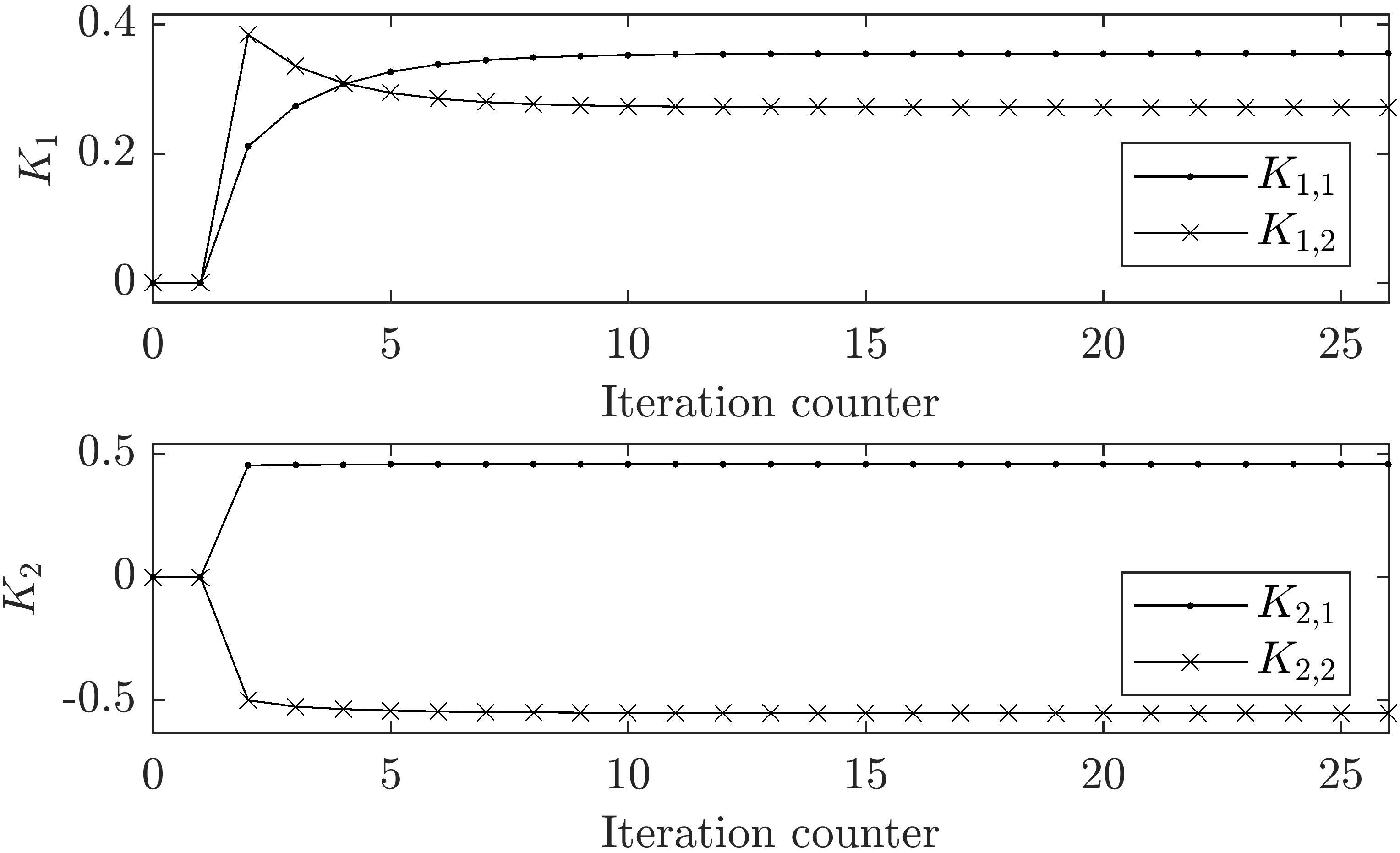}
    \caption{Convergence of Controller gains during learning.}
    \label{fig:3}
\end{figure}

\section{Conclusion}\label{section:conclusion}
In this research, we use RL concepts to develop a model-free controller for discrete-time MJLS based on the data exchanged between the unknown plant and intelligent agent. It is proved that the model-free controller converges to the model-based optimal controller. It is also shown that inclusion of excitation noise does not lead to any bias during the learning process.  Therefore, the stability of closed-loop MJLS under the proposed model-free controller is guaranteed. Finally, simulation study is carried out to evaluate the performance of the proposed controller. This research can be extended in several ways, including model-free controller design for MJLS with partially measured state or mode variables. 
\bibliographystyle{apalike}       
\bibliography{autosam}   
\end{document}